\documentclass{svjour3}                     % onecolumn (standard format)
\smartqed  % flush right qed marks, e.g. at end of proof
\usepackage{amsmath,graphicx,bm,latexsym,amssymb}
\usepackage{gensymb,algorithm,algorithmic}
\usepackage{color}
\newcommand{\val}{\boldsymbol{\alpha}}
\newcommand{\vbe}{\boldsymbol{\beta}}
\newcommand{\vmu}{\boldsymbol{\mu}}
\newcommand{\vtmu}{\tilde{\boldsymbol{\mu}}}
\newcommand{\vc}{\boldsymbol{c}}
\newcommand{\vu}{\boldsymbol{u}}
\newcommand{\vtu}{\tilde{\boldsymbol{u}}}
\newcommand{\vh}{\boldsymbol{h}}
\newcommand{\vth}{\tilde{\boldsymbol{h}}}
\newcommand{\vw}{\boldsymbol{w}}
\newcommand{\vv}{\boldsymbol{v}}
\newcommand{\vb}{\boldsymbol{b}}
\newcommand{\vone}{\boldsymbol{1}}
\newcommand{\vzero}{\boldsymbol{0}}
\newcommand{\veta}{\boldsymbol{\eta}}
\newcommand{\vpsi}{\boldsymbol{\psi}}
\newcommand{\vtpsi}{\tilde{\boldsymbol{\psi}}}
\newcommand{\vp}{\boldsymbol{p}}
\newcommand{\vg}{\boldsymbol{g}}
\newcommand{\vphi}{\boldsymbol{\phi}}
\newcommand{\vz}{\boldsymbol{z}}
\newcommand{\vx}{\boldsymbol{x}}
\newcommand{\vga}{\boldsymbol{\gamma}}

\newcommand{\mU}{\boldsymbol{U}}
\newcommand{\mtU}{\tilde{\boldsymbol{U}}}
\newcommand{\mV}{\boldsymbol{V}}
\newcommand{\mSig}{\boldsymbol{\Sigma}}
\newcommand{\mtSig}{\tilde{\boldsymbol{\Sigma}}}
\newcommand{\mLam}{\boldsymbol{\Lambda}}
\newcommand{\mI}{\boldsymbol{I}}
\newcommand{\mA}{\boldsymbol{A}}
\newcommand{\mC}{\boldsymbol{C}}
\newcommand{\mtC}{\tilde{\boldsymbol{C}}}
\newcommand{\mPhi}{\boldsymbol{\Phi}}
\newcommand{\mX}{\boldsymbol{X}}

\newcommand{\mB}{\boldsymbol{B}}
\newcommand{\mF}{\boldsymbol{F}}

\newcommand{\mY}{\boldsymbol{Y}}
\newcommand{\mDel}{\boldsymbol{\Delta}}
\newcommand{\mZ}{\boldsymbol{Z}}
\newcommand{\mtZ}{\tilde{\boldsymbol{Z}}}
\newcommand{\mGa}{\boldsymbol{\Gamma}}

\newcommand{\bR}{\mathbb{R}}
\newcommand{\cM}{\mathcal{M}}
\newcommand{\cV}{\mathcal{V}}
\newcommand{\cN}{\mathcal{N}}
\newcommand{\Pro}{\mathsf{P}}
\newcommand{\Exp}{\mathsf{E}}
\newcommand{\Tr}{\mathsf{Tr}}
\newcommand{\lse}{\text{lse}}
\newcommand{\diag}{\text{diag}}

\newcommand{\be}{\begin{equation}}
\newcommand{\ee}{\end{equation}}
\newcommand{\ba}{\begin{align}}
\newcommand{\ea}{\end{align}}
\newcommand{\nn}{\nonumber}

\journalname{Journal of Signal Processing Systems}

\begin{document}

\title{Multimodal Event Detection in Twitter Hashtag Networks
\thanks{This work was funded in part by the Consortium for Verification Technology under Department
of Energy National Nuclear Security Administration
award number DE-NA0002534, and the Army Research
Office (ARO) under grants W911NF-11-1-0391 and W911NF-12-1-0443.}
}


\author{Yasin Y{\i}lmaz         \and
        Alfred O. Hero %etc.
}


\institute{Y. Y{\i}lmaz \at
              Department of Electrical Engineering and Computer Science, University of Michigan, Ann Arbor, MI 48109, USA \\
              Tel.: +1-734-763-5022\\
              \email{yasiny@umich.edu}           %  \\
           \and
           A. Hero \at
              Department of Electrical Engineering and Computer Science, University of Michigan, Ann Arbor, MI 48109, USA
}

\date{Received: date / Accepted: date}

\maketitle

\begin{abstract}
Event detection in a multimodal Twitter dataset is considered. We treat the hashtags in the dataset as instances with two modes: text and geolocation features. The text feature consists of a bag-of-words representation. The geolocation feature consists of geotags (i.e., geographical coordinates) of the tweets. Fusing the multimodal data we aim to detect, in terms of topic and geolocation, the interesting events and the associated hashtags. To this end, a generative latent variable model is assumed, and a generalized expectation-maximization (EM) algorithm is derived to learn the model parameters. The proposed method is computationally efficient, and lends itself to big datasets. Experimental results on a Twitter dataset from August 2014 show the efficacy of the proposed method. 

\keywords{Event detection \and Twitter hashtag networks \and Multimodal data fusion \and Generative latent variable model \and Variational EM algorithm}
\end{abstract}

\section{Introduction}
\label{intro}

Twitter is the most popular microblogging service and the second most popular social network with over 300 million active users generating more than 500 million tweets per day as of 2015 \cite{Stats}. Its user-generated content from all over the world provides a valuable source of data for researchers from a variety fields such as machine learning, data mining, natural language processing, as well as social sciences. Twitter data has been used for various tasks, e.g., event detection \cite{Farzindar15}, sentiment analysis \cite{Liu12}, breaking news analysis \cite{Amer12}, rumor detection \cite{Zhao15}, community detection \cite{Brandon15}, election results prediction \cite{Tumasjan10}, and crime prediction \cite{Wang12}.

Hashtags, which are keywords preceded by the hash sign \#, are in general used to indicate the subject of the tweets. Hence, they provide useful information for clustering tweets or users. However, it is a noisy information source since hashtags are generated by users, and sometimes convey inaccurate or even counterfactual information. A small percentage of users (around 2\%) also geotag their tweets. Given the 500 million tweets per day, geotags also constitute an important information source. 

The detection of real-world events from conventional media sources has long been studied \cite{Yang98}. Event detection in Twitter is especially challenging because tweets use microtext, which is an informal language with a preponderance of abbreviated words, spelling and grammar errors. There are also many tweets of dubious value, consisting of nonsense, misrepresentations, and rumors. Much of the work on event detection in Twitter has considered a diversity of event types. For instance, \cite{Phuvipa10} considers unsupervised breaking news detection; \cite{Popescu11} considers supervised detection of controversial news events about celebrities; \cite{Benson11} addresses supervised musical event detection; and \cite{Sakaki10} deals with supervised natural disaster events monitoring. There are also a significant number of papers that consider unsupervised detection of events that do not require prespecification of the event type of interest, e.g., \cite{Petrovic10,Becker11,Long11,Weng11,Cordeiro12}.

In this paper, we introduce a new unsupervised event detection approach to Twitter that exploits the multimodal nature of the medium. Data is pre-processed to form a network of hashtags. In this network, each unique hashtag is an instance with multimodal features, namely text and geolocation. For a hashtag, the text feature is given by the bag-of-words representation over the collection of words from tweets that use the hashtag. The geolocation feature of a hashtag consists of the geotags of the tweets that mention the hashtag. The proposed approach can detect events in terms of both topic and geolocation through multimodal data fusion. To fuse the multimodal data we use a probabilistic generative model, and derive an expectation-maximization (EM) algorithm to find the maximum likelihood (ML) estimates of the model parameters. The proposed model can be seen as a multimodal factor analysis model \cite{mmfa,Khan10}. However, it is more general than the model in \cite{Khan10} in terms of the considered probabilistic models, and also the temporal dimension that is inherent to our problem. 

Fusing disparate data types, such as text and geolocation in our case, poses significant challenges. In \cite{Adali15}, source separation is used to fuse multimodal data, whereas \cite{Bramon12} follows an information-theoretic approach. Multimodal data fusion is studied for different applications such as multimedia data analysis \cite{Wu04} and brain imaging \cite{Sui12}. Multimodal feature learning via deep neural networks is considered in \cite{Ngiam11}. The literature on multi-view learning, e.g., \cite{Chris08,He11,Sun13}, is also related to the problem of multimodal data fusion. Our contributions in this paper are twofold. Firstly, we propose a intuitive framework that naturally extends to the exponential family of distributions. Secondly, based on a simple generative model, the proposed algorithm is computationally efficient, and thus applicable to big datasets. 

The paper is organized as follows. In Section \ref{sec:prob}, we formulate the multimodal event detection problem, and propose a generative latent variable model. Then, a generalized EM algorithm is derived in Section \ref{sec:EM}. Finally, experiment results on a Twitter dataset are presented in Section \ref{sec:exp}, and the paper is concluded in Section \ref{sec:conc}. We represent the vectors and matrices with boldface lowercase and uppercase letters, respectively. 

%Multimodal data fusion: factor analysis, refs, advantages, uses

\section{Problem Formulation}
\label{sec:prob}

\subsection{Observation Model}
\label{sec:obs}

We consider $P$ hashtags with text (i.e., collection of words used in tweets) and geotag (i.e., user geolocation) features, as shown in Table \ref{tab:htag}.

\begin{table}[htb]
\caption{Sample hashtags with text and geotag features.}
\label{tab:htag}
\begin{tabular}{| p{0.2\linewidth} | p{0.3\linewidth} | p{0.37\textwidth} |}
\hline
{\bf Hashtag} & {\bf Text} & {\bf Geotag (Latitude, Longitude)}  \\ 
\hline
\#Armstrong & \#Oprah mag 'alles vragen' aan Lance \#Armstrong. Uiteraard! Looking forward to the \#Lance \#Armstrong interview next week! \ldots & (52.4\degree N, 4.9\degree E) \newline (43.5\degree N, 79.6\degree W) \newline {\centering \ldots} \\
\hline
\#Arsenal & Sementara menunggu Team Power beraksi..\#Arsenal First game of 2013, lets start it off with a our fifth win in a row! Come on you Gunners! \#Arsenal  & (8.6\degree S, 116.1\degree E) \newline (23.7\degree N, 58.2\degree E) \newline \ldots \\
\hline
\end{tabular}
\end{table}

We assume a model in which each word in a tweet that uses the $i$-th hashtag is independently generated from the multinomial distribution with a single trial (i.e., categorical distribution) $\cM(1; p_{i1},\ldots,p_{iD})$, where $p_{id}$ is the probability of the $d$-th word for the $i$-th hashtag, and $D$ is the dictionary size. In this model, the word counts $\vh_i=[h_{i1},\ldots,h_{iD}]^T$ for the $i$-th hashtag are modeled as
\be
	\vh_i \sim \cM(M_i;p_{i1},\ldots,p_{iD}), ~i=1,\ldots,P, \nn
\ee
where $M_i$ is the number of dictionary words used in the tweets for the $i$-th hashtag, i.e., $M_i=\sum_{d=1}^D h_{id}$. To this end, we use the bag of words representation for the hashtags (Fig. \ref{fig:bow}). 

\begin{figure}[htb]
\includegraphics[width=\textwidth]{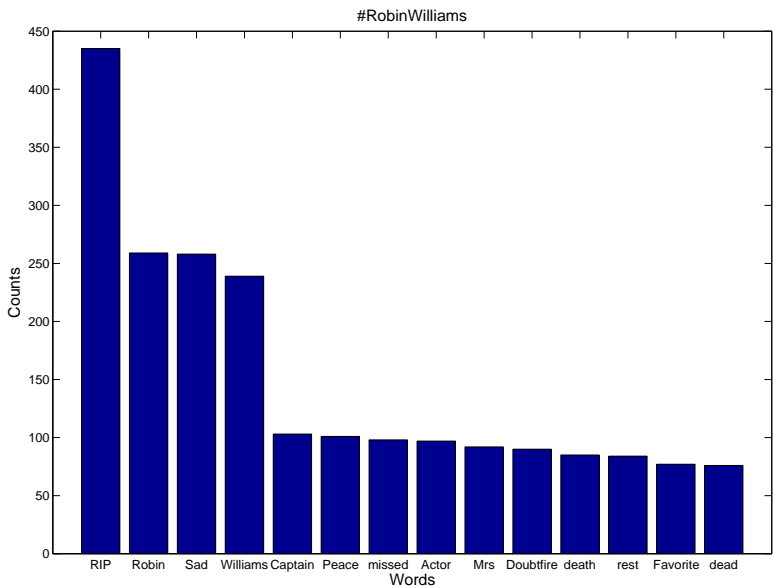}
\caption{A sample bag of words representation for the hashtag \#RobinWilliams.}
\label{fig:bow}       % Give a unique label
\end{figure}

The geolocation data of each tweet is a geographical location represented by a spherical coordinate (latitude and longitude). This coordinate is modeled using the 3-dimensional von Mises-Fisher (vMF) distribution, which is an extension of the Gaussian distribution to the unit sphere \cite{Mardia00} (Fig. \ref{fig:vmf}). We first convert the geographical coordinates (latitude, longitude) to the Cartesian coordinates ($x,y,z$), where $x^2+y^2+z^2=1$. Specifically, in our model, it is assumed that the geolocation of the $n$-th tweet that mentions the $i$-th hashtag is generated independently from the other tweets as follows
\be
	\vw_{in} \sim \cV(\val_i,\kappa_i),  ~i=1,\ldots,P,~n=1,\ldots,N_i, \nn
\ee
where $\val_i \in \bR^3,~\val_i^T\val_i=1,$ is the mean direction, $\kappa_i \ge 0$ is the concentration parameter, and $N_i$ is the number of geotagged tweets for the $i$-th hashtag. Larger $\kappa_i$ means more concentrated distribution around $\val_i$. Therefore, a local hashtag, such as \#GoBlue, which is used by supporters of the University of Michigan sports teams, requires a large $\kappa$, whereas a global hashtag, such as \#HalaMadrid, which means ``Go Madrid" and is used by the fans of the Real Madrid soccer team, requires a small $\kappa$ (Fig. \ref{fig:geo}). This difference in $\kappa$ is due to the fact that the Real Madrid supporters are well distributed around the globe, while the University of Michigan supporters are mostly confined to North America. 

\begin{figure}[htb]
\includegraphics[width=\textwidth]{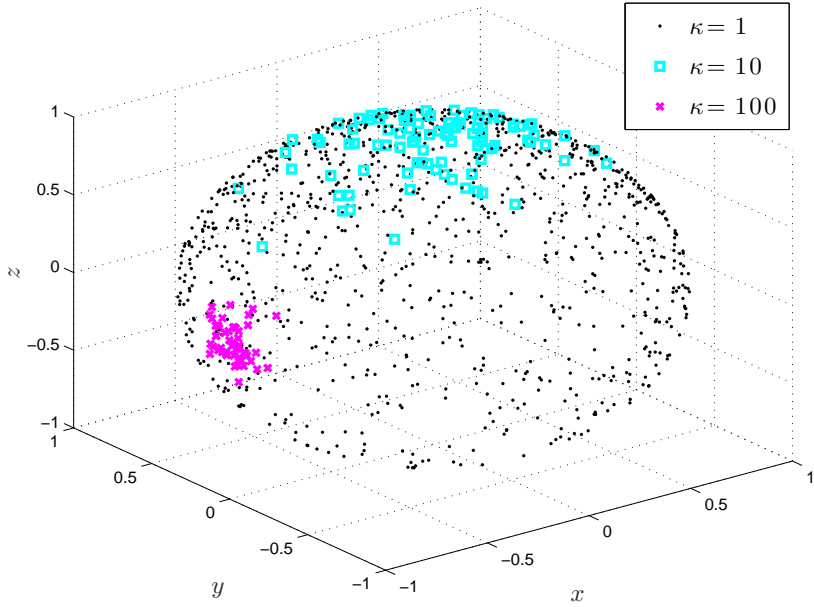}
\caption{Samples from the 3-dimensional von Mises-Fisher distribution with different concentration parameter values $\kappa=1, 10, 100$ describing the spread of the distribution around random mean directions. The case $\kappa=1$ produces the uniform distribution on the sphere.}
\label{fig:vmf}       % Give a unique label
\end{figure}

\begin{figure}[htb]
\includegraphics[width=\textwidth]{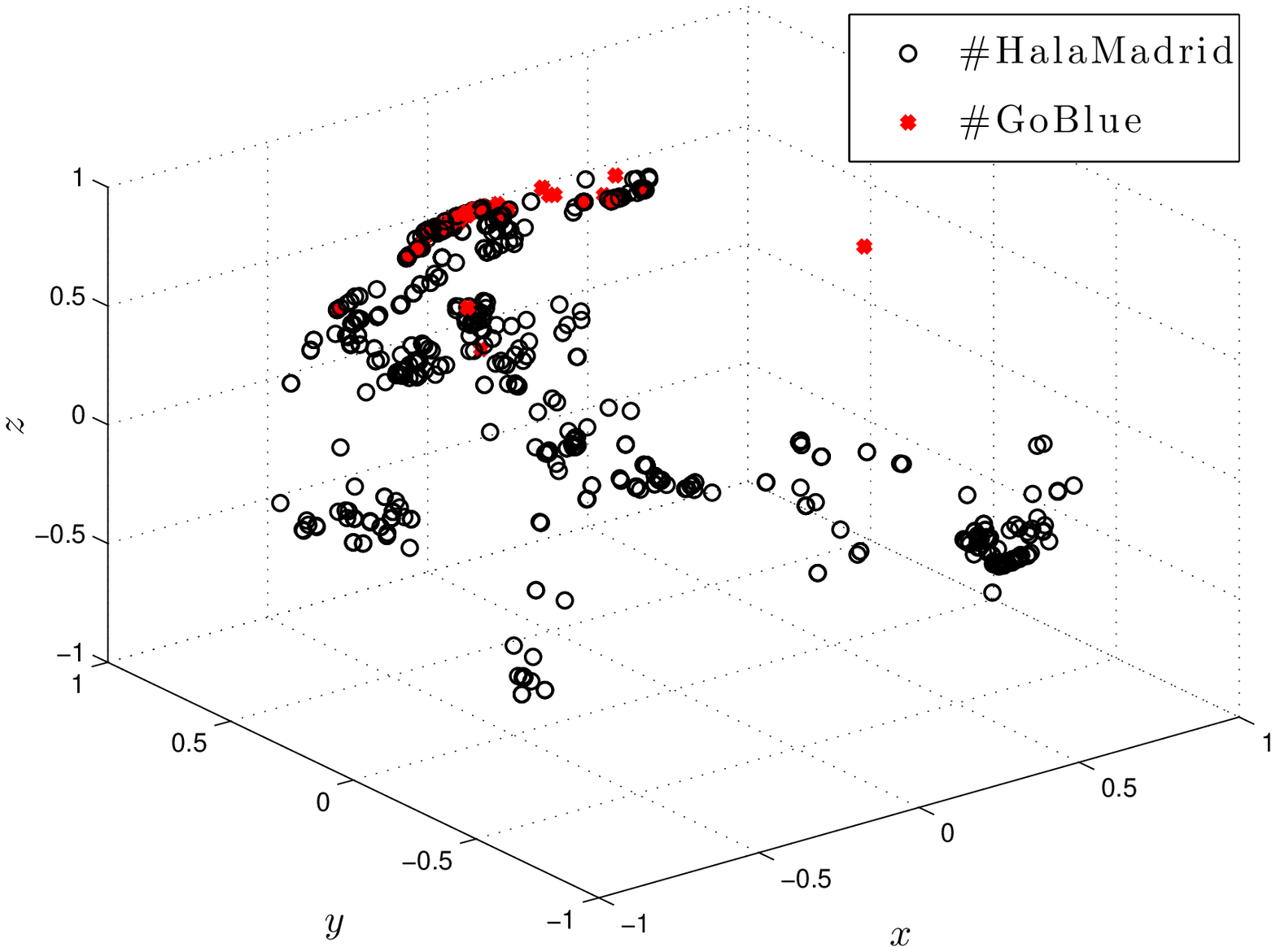}
\caption{Geolocations for the hashtags \#HalaMadrid (used for the Real Madrid soccer team) and \#GoBlue (used for the University of Michigan athletics) in terms of the Cartesian coordinates. The estimated concentration parameters for the von Mises-Fisher distribution are $\kappa_{\text{madrid}}=1.3302$ and $\kappa_{\text{mich}}=44.6167$, representing the wider global interest in Real Madrid soccer team as contrasted to the US-centric interest in University of Michigan sports teams.}
\label{fig:geo}       % Give a unique label
\end{figure}

\subsection{Generative Latent Variable Model}
\label{sec:gen}

Some hashtags are created by users as a result of an underlying event in {\em time} and {\em space}, which we call a {\em generative event}. For instance, after Robin Williams' death, many hashtags such as \#RobinWilliams, \#RIPRobinWilliams, \#RIPRobin, \#mrsdoubtfire have been used to commemorate him. On the other hand, some hashtags are more spread out over time, such as \#jobs, \#love, \#Healthcare, \#photo. With a slight abuse of the terminology, we also consider such an underlying topic as a {\em generative event}. In addition to the topic/text feature, a generative event (time-dependent or -independent) possesses also a spatial feature due to the event's geolocation (e.g., Asia, America) or simply due to the language (e.g., English, Spanish). 

We know that an event can generate multiple hashtags. Although there is usually a single event responsible for the generation of a hashtag, for generality, we let multiple events contribute to a single hashtag. In our generative model, $K$ events linearly mix in the {\em natural parameters} of the multinomial and vMF distributions to generate the text and geolocation features of each hashtag, respectively. Let $\vc_i \in \bR_+^K$ denote the mixture coefficients of $K$ events for the $i$-th hashtag, where $\bR_+$ is the set of nonnegative real numbers. Also let 
\be
	\mU = \left[ \vu_1 \ldots \vu_K \right] = \left[ \vu_{(1)}^T \ldots \vu_{(D)}^T \right]^T, ~ \vu_k \in \bR^D, ~ \vu_{(d)} \in \bR^{1\times K}, \nn
\ee
denote the event scores for the words in the dictionary; and 
\be
	\mV = [\vv_1 \ldots \vv_K],~ \vv_k \in \bR^3, \nn
\ee
denote the event geolocations in the Cartesian coordinates. Then, in our model, the mean of the vMF distribution is given by the normalized linear mixture 
\be
	\val_i = \frac{\mV \vc_i}{\|\mV \vc_i\|}, ~~i=1,\ldots,P,  \nn
\ee
where $\|\cdot\|$ is the $l^2$-norm normalization is required to ensure that $\val_i$ is on the unit sphere;
and the multinomial probabilities are given by the softmax function of the linear mixture $\vu_{(d)} \vc_i$, i.e., 
\be
	p_{id} = \frac{e^{\vu_{(d)} \vc_i}}{\sum_{j=1}^D e^{\vu_{(j)} \vc_i}}, ~~i=1,\ldots,P, ~~d=1,\ldots,D. \nn
\ee
That is,
\begin{align}
	\label{eq:mult_like}
	\vh_i &\sim \cM\left( M_i; \frac{e^{\vu_{(1)} \vc_i}}{\sum_{j=1}^D e^{\vu_{(j)} \vc_i}},\ldots,\frac{e^{\vu_{(d)} \vc_i}}{\sum_{j=1}^D e^{\vu_{(j)} \vc_i}} \right), ~~i=1,\ldots,P \\
	\label{eq:vmf_like}
	\vw_{in} &\sim \cV\left( \frac{\mV \vc_i}{\|\mV \vc_i\|}, \kappa_i \right),  ~~i=1,\ldots,P, ~~n=1,\ldots,N_i,
\end{align}
where $N_i$ is the number of geotagged tweets for the $i$-th hashtag.
We assume a Gaussian prior for the latent variable vector $\vu_k \in \bR^D$
\be
	\label{eq:mult_prior}
	\vu_k \sim \cN(\vmu_k,\mSig_k), ~~ k=1,\ldots,K,
\ee
and a vMF prior for $\vv_k \in \bR^3$
\be
	\label{eq:vmf_prior}
	\vv_k \sim \cV(\vbe_k,s_k), ~~ k=1,\ldots,K,
\ee
since the conjugate prior to the vMF likelihood with unknown mean and known concentration is also vMF \cite{Mardia76}.

The graphical model in Fig. \ref{fig:graph} depicts the proposed generative latent variable model. The proposed model can be regarded as a {\em multimodal factor analysis} model \cite{mmfa} since it combines features from two disparate domains(geotag and text). In classical factor analysis \cite{Harman76}, the mean of a Gaussian random variable is modeled with the linear combination $\vc^T\vu$ of factor scores in $\vu$, where the coefficients in $\vc$ are called factor loadings.The number of factors is typically much less than the number of variables modeled, as $K \ll P$ in our case. In the proposed model, the generative events correspond to the factors with the multimodal scores $\{\vu_k\}$ and $\{\vv_k\}$ for the multinomial and vMF observations, respectively. For both modalities, the natural parameters are modeled with the linear combination of the factor scores using the {\em same factor loading vector} $\vc_i$ for the $i$-th hashtag. In the multinomial distribution, the softmax function maps the natural parameters to the class probabilities, whereas in the vMF distribution, the natural parameter coincides with the (scaled) mean. For each hashtag $i$, the factor loading vector $\vc_i$ correlates the two observation modalities: text and geolocation. 

\begin{figure}[tb]
\includegraphics[width=\textwidth]{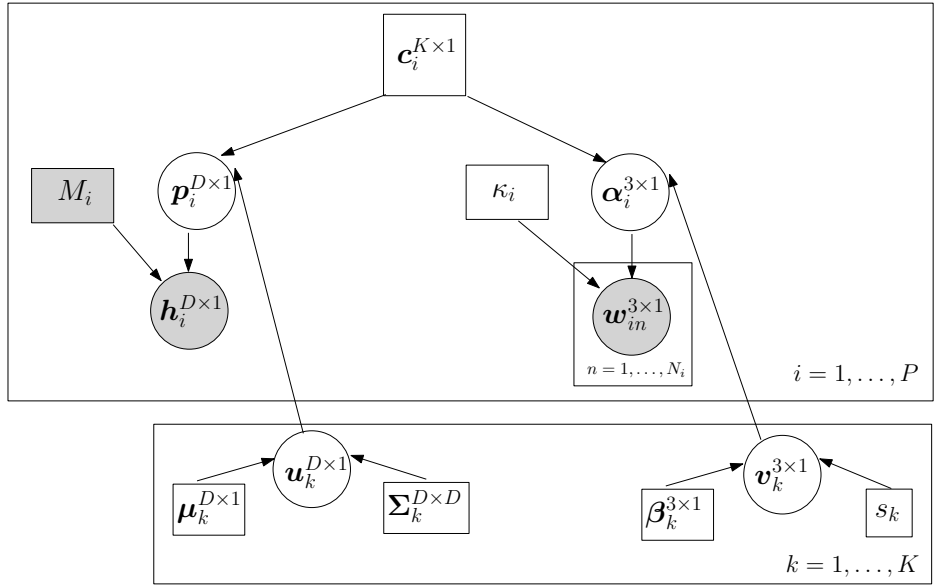}
\caption{Generative graphical model. Plate representation is used to show repeated structures. Circles and rectangles represent random and deterministic variables, respectively. Observed variables are shaded.}
\label{fig:graph}       % Give a unique label
\end{figure}

Next we will present an EM algorithm to learn the parameters of the proposed model from the data. 

\section{EM Algorithm}
\label{sec:EM}

We propose a generalized EM (GEM) algorithm that consists of two separate EM steps for the two modalities, and a coordinating M step for the mixture coefficients $\{\vc_i\}$. Specifically, at each iteration of the GEM algorithm, the vMF EM steps are followed by the multinomial EM steps, which are followed by the M step for $\{\vc_i\}$. The individual EM steps for vMF and multinomial are coupled only through $\{\vc_i\}$, and independent otherwise. In the proposed GEM algorithm, the global likelihood function is monotonically increasing. 

\subsection{Von Mises-Fisher Parameters}
\label{sec:vmf}

We would like to find the ML estimates of the parameters $\vbe_k,s_k,\kappa_i$ under the graphical model depicted in the right branch of Fig. \ref{fig:graph}. We take a variational EM approach to deal with the latent variable vectors $\{\vv_k\}$. 

\subsubsection{E-Step}

Starting with the E-step we seek the posterior probability density function (pdf) $\Pro(\{\vv_k\} | \{\vw_{in}\},\theta)$, where $\theta=\{\vbe_k,s_k,\kappa_i, \vc_i\}$. From \eqref{eq:vmf_like} and \eqref{eq:vmf_prior}, we know that the likelihood $\Pro(\{\vw_{in}\} | \{\vv_k\},\theta )$ and the prior $\Pro(\{\vv_k\} | \theta)$ are both vMF, hence the joint distribution is given by
\begin{align}
	\Pro(\{\vv_k\}, \{\vw_{in}\} | \theta) &= \Pro(\{\vw_{in}\} | \{\vv_k\},\theta) ~ \Pro(\{\vv_k\} | \theta) \nn\\
	&= \prod_{i=1}^P C(\kappa_i)^{N_i} \prod_{n=1}^{N_i}  \exp\left( \kappa_i \vw_{in}^T \frac{\mV \vc_i}{\|\mV \vc_i\|} \right) \prod_{k=1}^K C(s_k) \exp\left( s_k \vv_k^T \vbe_k \right) \nn,
\end{align}
where
\be
	\label{eq:C_def}
	C(x) = \frac{x^{1/2}}{(2\pi)^{3/2} I_{1/2}(x)} = \frac{x}{2\pi(e^{x}-e^{-x})}
\ee
is the normalization factor in the 3-dimensional vMF pdf, with $I_y(x)$ being the modified Bessel function of the first kind at order $y$. 
Reorganizing the terms we get
\begin{align}	
	\Pro(\{\vv_k\}, \{\vw_{in}\} | \theta) &= \prod_{i=1}^P C(\kappa_i)^{N_i} \prod_{k=1}^K C(s_k) ~\nn\\
	& \quad \quad \exp\left( \sum_{i=1}^P \sum_{n=1}^{N_i} \kappa_i \vw_{in}^T \sum_{k=1}^K \frac{c_{ik}}{\|\mV \vc_i\|} \vv_k + \sum_{k=1}^K s_k \vv_k^T \vbe_k \right) \nn\\
	\label{eq:vmf_E}
	&= \prod_{i=1}^P C(\kappa_i)^{N_i} \prod_{k=1}^K C(s_k) ~\nn\\
	& \quad \quad \prod_{k=1}^K \exp\left( \vv_k^T \left( \sum_{i=1}^P \sum_{n=1}^{N_i} \frac{c_{ik}}{\sqrt{\vc_i^T \mV^T\mV \vc_i}} \kappa_i \vw_{in} + s_k \vbe_k \right) \right).
\end{align}
In the alternative expression for the joint pdf
\be
	\Pro(\{\vv_k\}, \{\vw_{in}\} | \theta) = \Pro(\{\vv_k\} | \{\vw_{in}\},\theta) ~ \Pro(\{\vw_{in}\} | \theta), \nn
\ee
the dependency on $\{\vv_k\}$ appears only in the posterior pdf, hence $\Pro(\{\vv_k\} | \{\vw_{in}\},\theta)$ lies in the exponential term in \eqref{eq:vmf_E}, which resembles the vMF pdf except the dependence of the normalization factor on $\{\vv_k\}$. The diagonal entries of $\mV^T\mV$ are $\vv_k^T\vv_k=1$; and the off-diagonal entries are $\vv_j^T\vv_k \le 1,~ j\not=k$. Since $c_{ik}\ge0,~k=1,\ldots,K$, the inequality $\vc_i^T \mV^T\mV \vc_i \le \vc_i^T \vone_K\vone_K^T \vc_i$ holds, where $\vone_K$ is the vector of $K$ ones. To make \eqref{eq:vmf_E} tractable we replace $\vc_i^T \mV^T\mV \vc_i$ with $\vc_i^T \vone_K\vone_K^T \vc_i$ and obtain the lower bound
\begin{align}	
	\label{eq:vmf_lb}
	\Pro(\{\vv_k\}, \{\vw_{in}\} | \theta) &\ge Q_v(\{\vv_k\}, \theta) \nn\\
	&= \prod_{i=1}^P C(\kappa_i)^{N_i} \prod_{k=1}^K C(s_k) ~\nn\\
	& \quad \quad \prod_{k=1}^K \exp\left( \vv_k^T \left( \sum_{i=1}^P \sum_{n=1}^{N_i} \frac{c_{ik}}{\sqrt{\vc_i^T \vone_K\vone_K^T \vc_i}} \kappa_i \vw_{in} + s_k \vbe_k \right) \right).
\end{align}
To put \eqref{eq:vmf_lb} into the standard form of the vMF pdf we normalize the term in the inner parentheses and obtain 
\begin{align}
	\label{eq:vmf_post}
	Q_v(\{\vv_k\}, \theta) &= \prod_{i=1}^P \prod_{n=1}^{N_i} \prod_{k=1}^K \frac{C(\kappa_i)^{N_i} ~C(s_k)}{C(r_k)} ~\underbrace{C(r_k) \exp\left( r_k \vv_k^T \vb_k \right)}_{q_v(\vv_k)}, \\
	\label{eq:vmf_E_mean}
	\vb_k &= \frac{\sum_{i=1}^P \frac{c_{ik}}{\sum_{k=1}^K c_{ik}} \kappa_i \sum_{n=1}^{N_i} \vw_{in} + s_k \vbe_k}{\left\| \sum_{i=1}^P \frac{c_{ik}}{\sum_{k=1}^K c_{ik}} \kappa_i \sum_{n=1}^{N_i} \vw_{in} + s_k \vbe_k \right\|}, \\
	\label{eq:vmf_E_conc}
	r_k &= \left\| \sum_{i=1}^P \frac{c_{ik}}{\sum_{k=1}^K c_{ik}} \kappa_i \sum_{n=1}^{N_i} \vw_{in} + s_k \vbe_k \right\|,
\end{align}
where $\vb_k$ is the mean direction and $r_k$ is the concentration parameter.
We approximate the posterior $\Pro(\vv_k | \{\vw_{in}\},\theta)$ with the vMF distribution $q_v(\vv_k)$ for $k=1,\ldots,K$. 

\subsubsection{M-Step}

In the M-step, we find the parameters $\vbe_k,s_k,\kappa_i$ that maximize the expected value of the lower bound for the complete-data log-likelihood, which from \eqref{eq:vmf_lb} is given by 
\begin{multline}
\label{eq:vmf_M}
	\Exp_{q_v(\vv_k)} \left[ \log Q_v(\{\vv_k\}, \theta) \right] = \sum_{k=1}^K \left( \sum_{i=1}^P \frac{c_{ik}}{\sum_{k=1}^K c_{ik}} \kappa_i \sum_{n=1}^{N_i} \vw_{in} + s_k \vbe_k \right)^T \vb_k \\ + \sum_{k=1}^K \log C(s_k) + \sum_{i=1}^P N_i \log C(\kappa_i), 
\end{multline}
where the expectation is taken over $q_v(\vv_k)$, which approximates  the posterior pdf $\Pro(\vv_k | \{\vw_{in}\},\theta)$ (see \eqref{eq:vmf_post}). 

We start with the estimator $\hat{\kappa}_i$ which is given by
\be
	\hat{\kappa}_i = \arg\max_{\kappa_i} ~\kappa_i \left( \sum_{n=1}^{N_i} \vw_{in} \right)^T \sum_{k=1}^K \frac{c_{ik}}{\sum_{k=1}^K c_{ik}} \vb_k + N_i \log C(\kappa_i). \nn
\ee
{We next show that $\hat{\kappa}_i$ always makes the derivative with respect to $\kappa_i$ zero, i.e.,} 
\be
	\label{eq:tau}
	-\frac{C'(\hat{\kappa}_i)}{C(\hat{\kappa}_i)} = \left( \frac{1}{N_i} \sum_{n=1}^{N_i} \vw_{in} \right)^T \sum_{k=1}^K \frac{c_{ik}}{\sum_{k=1}^K c_{ik}} \vb_k \triangleq \tau_i. 
\ee
From \eqref{eq:C_def}, we write the derivative $C'(\hat{\kappa}_i)$ as
\begin{align}
	C'(\hat{\kappa}_i) = \frac{\hat{\kappa}_i^{1/2}}{(2\pi)^{3/2} I_{1/2}(\hat{\kappa}_i)} \left( \frac{1}{2\hat{\kappa}_i} - \frac{I_{1/2}'(\hat{\kappa}_i)}{I_{1/2}(\hat{\kappa}_i)} \right) = C(\hat{\kappa}_i) \left( \frac{1}{2\hat{\kappa}_i} - \frac{I_{1/2}'(\hat{\kappa}_i)}{I_{1/2}(\hat{\kappa}_i)} \right). \nn
\end{align}
Hence, 
\be
	-\frac{C'(\hat{\kappa}_i)}{C(\hat{\kappa}_i)} = \frac{\hat{\kappa}_i I_{1/2}'(\hat{\kappa}_i) - I_{1/2}(\hat{\kappa}_i)/2}{\hat{\kappa}_i I_{1/2}(\hat{\kappa}_i)}. \nn
\ee
Using \eqref{eq:tau} and the following recurrence relation \cite[Section 9.6.26]{Abram72}
\be
	x I_{3/2}(x) = x I_{1/2}'(x) - I_{1/2}(x)/2 \nn
\ee
we get
\be
	\label{eq:Bessel}
	\frac{I_{3/2}(\hat{\kappa}_i)}{I_{1/2}(\hat{\kappa}_i)} = \tau_i,
\ee
{which can be rewritten as \cite[Section 10.2.13]{Abram72}
\be
	\label{eq:Bessel2}
	\coth(\hat{\kappa}_i) -\frac{1}{\hat{\kappa}_i} = \tau_i.
\ee
We can also obtain \eqref{eq:Bessel2} using \eqref{eq:C_def} and \eqref{eq:tau}.
Fig. \ref{fig:derivative} shows that the left side of \eqref{eq:Bessel2} is a continuous function taking values in $[0,1]$. Since $\tau_i$, defined in \eqref{eq:tau}, is also in $[0,1]$, there always exists a unique solution to \eqref{eq:Bessel2}.
However, there is no analytical solution to \eqref{eq:Bessel} or \eqref{eq:Bessel2}; hence we resort to approximating $\hat{\kappa}_i$. A numerical solution easily follows using a root-finding method such as the bisection method.
\begin{figure}[tb]
\includegraphics[width=.5\textwidth]{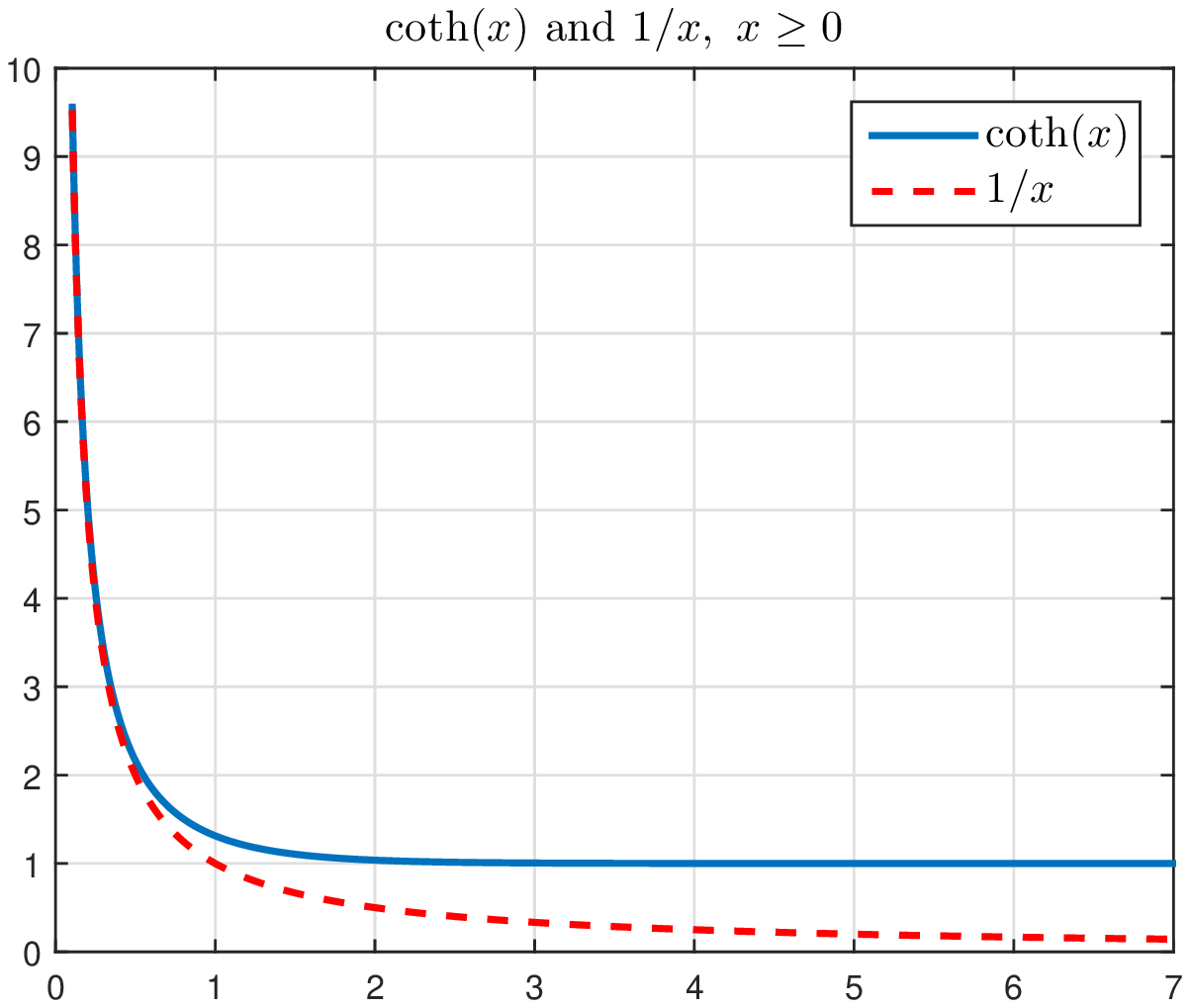}
\includegraphics[width=.5\textwidth]{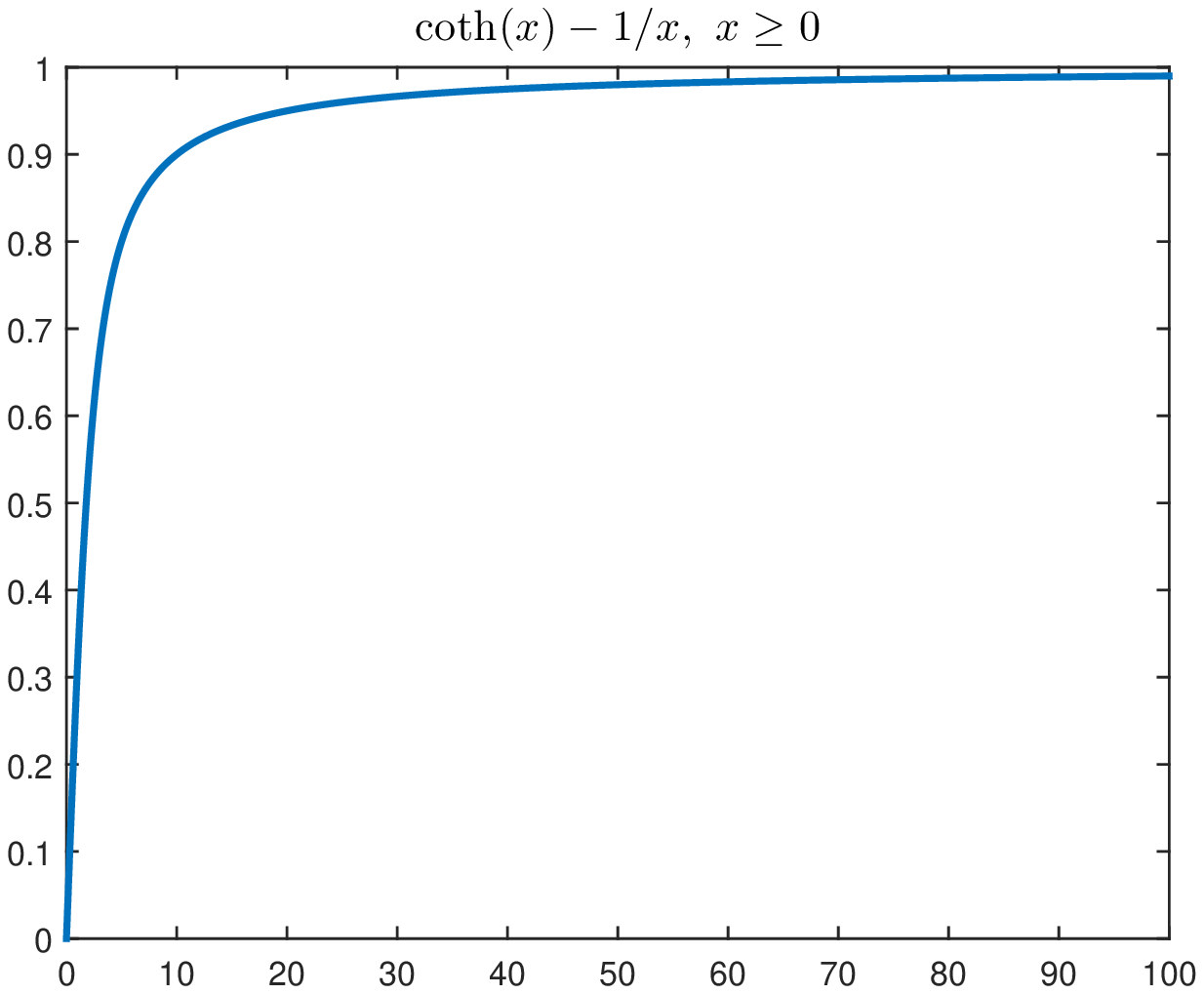}
\caption{For nonnegative numbers, the hyperbolic cotangent and the multiplicative inverse (left figure); and their difference (right figure).}
\label{fig:derivative} 
\end{figure}}

Alternatively, using the following continuing fraction representation
\be	
	\frac{I_{3/2}(\hat{\kappa}_i)}{I_{1/2}(\hat{\kappa}_i)} = \frac{1}{\frac{3}{\hat{\kappa}_i} + \frac{1}{\frac{5}{\hat{\kappa}_i} + \cdots}} = \tau_i \nn
\ee
we can approximate $\hat{\kappa}_i$ as \cite{Banerjee05}
\begin{align}
	\frac{1}{\tau_i} &\approx \frac{3}{\hat{\kappa}_i} + \tau_i \nn\\
	\hat{\kappa}_i &\approx \frac{3\tau_i}{1-\tau_i^2}. \nn
\end{align}
Furthermore, an empirical correction is also provided in \cite{Banerjee05}:
\be
	\label{eq:kappa}
	\hat{\kappa}_i \approx \frac{3\tau_i-\tau_i^3}{1-\tau_i^2},
\ee
which is constrained to be nonnegative for feasibility. We introduce a Lagrange multiplier $\lambda>0$, replacing $\tau_i$ with $\tilde{\tau}_i=\tau_i+\lambda$ to enforce this non-negativity constraint. Due to complementary slackness, this leads to the following estimator
\be
	\label{eq:kappa_est}
	\hat{\kappa}_i \approx \max\left\{ 0, \frac{3\tau_i-\tau_i^3}{1-\tau_i^2} \right\}.
\ee

Similar to $\kappa_i$ \eqref{eq:tau}--\eqref{eq:kappa_est}, from \eqref{eq:vmf_M}, we estimate $s_k$ with
\begin{align}
	\hat{s}_k &= \arg\max_{s_k} ~s_k \vbe_k^T \vb_k + \log C(s_k). \nn\\
	\label{eq:s_est}
	&\approx \max\left\{ 0, \frac{3\vbe_k^T \vb_k-(\vbe_k^T \vb_k)^3}{1-(\vbe_k^T \vb_k)^2} \right\}.
\end{align}

Since $\vbe_k$ is a mean direction on the unit sphere, it has to satisfy $\vbe_k^T \vbe_k =1$.Therefore, from \eqref{eq:vmf_M}, our estimator is given by
\be
	\hat{\vbe}_k = \arg\max_{\vbe} \vbe_k^T s_k \vb_k ~~~\text{subject to}~~~ \vbe_k^T \vbe_k =1. \nn
\ee
Maximum of $\vbe_k^T s_k \vb_k$ is attained when the angle between $\vbe_k$ and $s_k \vb_k$ is zero, i.e., $\hat{\vbe}_k=c ~s_k \vb_k$. Since the feasible set is the unit sphere, $\hat{\vbe}_k=\frac{s_k \vb_k}{\|s_k \vb_k\|}=\frac{\vb_k}{\|\vb_k\|}$. The posterior mean direction $\vb_k$, given by \eqref{eq:vmf_E_mean}, is already on the unit sphere, hence
\be
	\hat{\vbe}_k=\vb_k.
\ee

\subsection{Multinomial Parameters}
\label{sec:mult}

Note that there are $D-1$ degrees of freedom in the multinomial class probabilities due to the constraint $\sum_{d=1}^D p_{id}=1$. For identifiability, we set the $D$-th word as pivot, and deal with the latent event scores 
\be
	\vtu_{(d)} = \vu_{(d)}-\vu_{(D)},~~ d=1,\ldots,D-1, \nn
\ee
and accordingly $\mtU=[\vtu_1 \ldots \vtu_K]$, where from \eqref{eq:mult_prior}
\be
	\label{eq:mult_prior1}
	\vtu_k^{D-1 \times 1} \sim \cN(\vtmu_k, \mtSig_k).
\ee

\subsubsection{E-Step}

We seek the posterior pdf $\Pro\left( \{\vtu_k\} | \{\vh_i\},\theta \right)$ where $\theta=\{\vtmu_k,\mtSig_k,\vc_i\}$. From \eqref{eq:mult_like} and \eqref{eq:mult_prior1},
\begin{align}
	\label{eq:mult_post}
	\Pro\left( \{\vtu_k\}, \{\vh_i\} | \theta \right) &= \Pro\left( \{\vh_i\} | \{\vtu_k\},\theta \right) \Pro\left( \{\vtu_k\} | \theta \right) \nn\\
	&= \prod_{i=1}^P \frac{M_i!}{h_{i1}! \cdots h_{iD}!} \prod_{d=1}^D \exp\left( h_{id} \left[ \eta_{id}-\lse(\veta_i) \right] \right) \nn\\
	&\quad\quad \prod_{k=1}^K \frac{\exp\left( -\frac{1}{2} (\vtu_k-\vtmu_k)^T \mtSig_k^{-1} (\vtu_k-\vtmu_k) \right)}{(2\pi)^{(D-1)/2} ~|\mtSig_k|^{1/2}},
\end{align}
where $\eta_{id}=\vtu_{(d)}\vc_i,~d=1,\ldots,D-1,~\eta_{iD}=0, ~\veta_i=[\eta_{i1} \ldots \eta_{iD-1}]^T=\mtU\vc_i$, and the log-sum-exp function
\be
	\label{eq:lse}
	\lse(\veta_i)=\log\left( 1+\sum_{d=1}^{D-1} \exp\left( \eta_{id} \right) \right).
\ee
As in the vMF case \eqref{eq:vmf_E}, the normalization factor in \eqref{eq:mult_post}, which is the lse function, prevents a tractable form. Following \cite{Khan10} we use a quadratic upper bound of the lse function based on the Taylor series expansion to obtain a lower bound for the complete-data likelihood, given in \eqref{eq:mult_post}. The second order Taylor series expansion around a fixed point $\vpsi_i$ is given by
\be
	\lse(\veta_i) = \lse(\vpsi_i) + (\veta_i-\vpsi_i)^T \nabla\lse(\vpsi_i) + \frac{1}{2} (\veta_i-\vpsi_i)^T \nabla^2\lse(\vtpsi_i)(\veta_i-\vpsi_i), \nn
\ee
where $\tilde{\vpsi}_i \in (\veta_i,\vpsi_i)$. From \eqref{eq:lse}, 
\begin{align}
	\nabla\lse(\vpsi_i) &= \left[ \frac{\exp\left( \psi_{i1} \right)}{1+\sum_{d=1}^{D-1} \exp\left( \psi_{id} \right)} \cdots \frac{\exp\left( \psi_{iD-1} \right)}{1+\sum_{d=1}^{D-1} \exp\left( \psi_{id} \right)} \right] = \vp_{\vpsi_i} \nn\\
	\nabla\lse^2(\vtpsi_i) &= \mLam_{\vtpsi_i} - \vp_{\vtpsi_i}\vp_{\vtpsi_i}^T, ~~ \mLam_{\vtpsi_i} = \diag(\vp_{\vtpsi_i}), \nn
\end{align}
where $\mLam_{\vtpsi_i}$ is the diagonal matrix form of $\vp_{\vtpsi_i}$. In \cite{Bohning92}, it is shown that the matrix
\be
	\label{eq:A_def}
	\mA = \frac{1}{2} \left( \mI_{D-1}-\frac{\vone_{D-1}\vone_{D-1}^T}{D} \right) \succcurlyeq \nabla\lse^2(\vtpsi_i), ~~\forall \vtpsi_i, 
\ee
in the positive semi-definite sense, where $\mI_d$ is the $d$-dimensional identity matrix. That is,
\begin{align}
	\label{eq:lse_ub}
	\lse(\veta_i) &\le \frac{1}{2} \veta_i^T \mA \veta_i + \vg_{\vtpsi_i}^T \veta_i + c_{\vtpsi_i}, \\
	\vg_{\vtpsi_i} &= \vp_{\vpsi_i} - \mA\vpsi_i, \nn\\
	c_{\vtpsi_i} &= \lse(\vpsi_i) + \frac{1}{2} \vpsi_i^T \mA \vpsi_i - \vpsi_i^T \vp_{\vpsi_i}. \nn
\end{align}

In \eqref{eq:mult_post}, replacing $\lse(\veta_i)$ with the quadratic upper bound in \eqref{eq:lse_ub} we get the following lower bound for the likelihood $\Pro\left( \{\vh_i\} | \{\vtu_k\},\theta \right)$
\begin{align}
	\Pro\left( \{\vh_i\} | \{\vtu_k\},\theta \right) &\ge \prod_{i=1}^P \frac{M_i!}{h_{i1}! \cdots h_{iD}!} \nn\\ 
	&\quad\quad \exp\left( \sum_{d=1}^{D} h_{id} \eta_{id} - \left( \frac{1}{2} \veta_i^T \mA \veta_i + \vg_{\vtpsi_i}^T \veta_i + c_{\vtpsi_i} \right) \sum_{d=1}^D h_{id} \right) \nn\\
	&= \prod_{i=1}^P \frac{M_i!}{h_{i1}! \cdots h_{iD}!} \nn\\ 
	&\quad\quad \exp\left( - \frac{1}{2} \left( \veta_i^T M_i\mA \veta_i - 2M_i\left( \frac{\vh_{i\backslash D}}{M_i} - \vg_{\vtpsi_i} \right)^T \veta_i + 2M_i c_{\vtpsi_i} \right) \right), \nn	
\end{align}
where $\vh_{i\backslash D}=[h_{i1} \ldots h_{iD-1}]^T$ is the count vector of the $i$-th hashtag for the first $D-1$ words. Defining a new observation vector
\be
	\vth_i=\mA^{-1}\left( \frac{\vh_{i\backslash D}}{M_i} - \vg_{\vtpsi_i} \right) = \mA^{-1}\left( \frac{\vh_{i\backslash D}}{M_i} - \vp_{\vpsi_i} \right) + \vpsi_i \nn
\ee
we write
\begin{align}
	\label{eq:mult_mix}
	\Pro\left( \{\vth_i\} | \{\vtu_k\},\theta \right) &\ge \prod_{i=1}^P f_{\vtpsi_i} \exp\left( - \frac{1}{2} \left( \veta_i - \vth_i \right)^T M_i\mA \left( \veta_i - \vth_i \right) \right), \\
	f_{\vtpsi_i} &= \frac{M_i!}{h_{i1}! \cdots h_{iD}!} ~\exp\left( \frac{\vth_i^T M_i\mA \vth_i}{2} - M_i c_{\vtpsi_i} \right). \nn
\end{align}

Recall that $\veta_i=\mtU\vc_i=\sum_{k=1}^K c_{ik}\vtu_k$. In \eqref{eq:mult_mix}, the latent variable vectors $\{\vtu_k\}$, which are independently modeled {\it a priori} \eqref{eq:mult_prior1}, are coupled, thus no more independent {\it a posteriori} in $\Pro\left( \{\vtu_k\} | \{\vth_i\},\theta \right)$. To capture the dependency we treat them in a single vector $\vtu = [\vtu_{(1)} \ldots \vtu_{(K)}]^T$. Without loss of generality, {\it a priori} we assume independence among the words for the same event, i.e., $\mtSig_k=\mI_{D-1}, ~\forall k$. Prior probability distribution reflects our initial belief about the unknown entity; and {\it a priori} we do not know anything about the inter-word dependencies of the hidden events. Hence, this is a quite reasonable assumption. In any case (under any prior assumption), we learn the posterior distribution for $\vtu$. For the same reason, without loss of generality, we also assume $\vtmu_k=\vzero_{D-1}, ~\forall k$, i.e., 
\be
	\label{eq:mult_prior2}
	\vtu \sim \cN(\vzero_{K(D-1)}, \mI_{K(D-1)}).
\ee
To rewrite \eqref{eq:mult_mix} in terms of $\vtu$ we note that $\veta_i=\mtC_i^T \vtu$ where 
\be
	\label{eq:C_def}
	\mtC_i= \mI_{D-1} \otimes \vc_i,
\ee
and $\otimes$ denotes the Kronecker product. 

Then, from \eqref{eq:mult_mix} and \eqref{eq:mult_prior2}, we approximate the complete-data likelihood with the following lower bound
\begin{align}
	\label{eq:mult_lb}
	\Pro\left( \{\vth_i\}, \vtu | \theta \right) &\ge Q_m(\vtu,\{\vc_i\}) \nn\\
	&= \prod_{i=1}^P \frac{f_{\vtpsi_i}}{(2\pi)^{K(D-1)/2} } \nn\\
	&\quad\quad \exp\left( - \frac{1}{2} \left[ \left( \mtC_i^T\vtu - \vth_i \right)^T M_i\mA \left( \mtC_i^T\vtu - \vth_i \right) + \vtu^T \vtu \right] \right) \nn\\
	&= \prod_{i=1}^P f_{\vtpsi_i} \exp\left( \left[\vphi^T \mPhi^{-1} \vphi-\vth_i^T M_i\mA \vth_i \right] \big/ 2 \right) ~|\mPhi|^{1/2} \nn\\
	&\quad\quad \underbrace{\exp\left( - \frac{1}{2} (\vtu-\vphi)^T \mPhi^{-1} (\vtu-\vphi) \right) \Big/ (2\pi)^{K(D-1)/2} |\mPhi|^{1/2}}_{q_m(\vtu)},
\end{align}
where using \eqref{eq:C_def}
\begin{align}
	\label{eq:mult_mean}
	\vphi &= \mPhi \sum_{i=1}^P M_i \mtC_i \mA \vth_i = \mPhi \sum_{i=1}^P \left( M_i \mA \vth_i \right) \otimes \vc_i , \\
	\label{eq:mult_cov}
	\mPhi &= \left( \sum_{i=1}^P M_i \mtC_i \mA \mtC_i^T + \mI_{K(D-1)} \right)^{-1} = \left( \sum_{i=1}^P M_i \mA \otimes \vc_i \vc_i^T + \mI_{K(D-1)} \right)^{-1}.
\end{align}
Using the lower bound in \eqref{eq:mult_lb} we approximate the posterior $\Pro\left( \vtu | \{\vth_i\},\theta \right)$ with $q_m(\vtu)$, which is $\cN(\vphi,\mPhi)$. 

Note that $K(D-1)$ can be very large due to the dictionary size $D$. As a result, it is, in general, not practical to perform the matrix inversion in \eqref{eq:mult_cov}. From the Matrix Inversion Lemma, it can be shown that 
\begin{align}
	\label{eq:Phi_fin}
	\mPhi &= \mI_{D-1} \otimes \mF^{-1} - \vone_{D-1}\vone_{D-1}^T \otimes \mDel, \\
	\mF &= \frac{1}{2} \mC \mLam_{M_i} \mC^T + \mI_{K}, \nn\\
	\mDel &= \mF^{-1}\mC \mY \mC^T\mF^{-1} \nn\\
	\mY &= -\frac{\mLam_{M_i}}{2D} - \frac{\mLam_{M_i}}{2D} \mC^T \left( \frac{\mF}{D-1} - \mC\frac{\mLam_{M_i}}{2D}\mC^T \right)^{-1} \mC\frac{\mLam_{M_i}}{2D}, \nn
\end{align}
where $\mC=[\vc_1 \ldots \vc_P]$, and $\mLam_{M_i}$ is the diagonal matrix whose entries are $M_1,\ldots,M_P$.
Using \eqref{eq:Phi_fin} we efficiently compute $\mPhi$ by only inverting $K \times K$ matrices.
Since, typically, the number of events is selected a small number, the proposed algorithm is now feasible for big datasets with large $P$ and $D$ (see Theorem \ref{thm:comp} and Section \ref{sec:exp}).

We can similarly simplify the computation of $\vphi$, given in \eqref{eq:mult_mean}. Define
\be
	\vz_i = M_i \mA\vth_i, \nn
\ee
and partition the posterior mean $\vphi$ of the $K(D-1)$ event-word scores into $D-1$ vectors of size $K$
\be
	\label{eq:X_def}
	\vphi = [\vx_1^T \ldots \vx_{D-1}^T]^T, ~~ \mX = [\vx_1 \ldots \vx_{D-1}].
\ee
We can efficiently compute $\mX$, which is nothing but a reorganized version of $\vphi$, as
\begin{align}
	\label{eq:phi_matrix}
	\mX &= \mF^{-1}\mC\mZ - \mDel\mC\mtZ, \\
	\mZ &= [\vz_1 \ldots \vz_P]^T, \nn\\
	\mtZ &= \mZ ~\vone_{D-1} \vone_{D-1}^T. \nn
\end{align}

\subsubsection{M-step}

The mean and covariance of $\vtu$ are estimated using \eqref{eq:mult_mean} and \eqref{eq:mult_cov}. From \cite{Khan10}, the optimum value of $\vpsi_i$ is given by
\be
	\vpsi_i = \mtC_i^T \vphi = \mX^T \vc_i.
\ee
For the estimation of $\{\vc_i\}$, which is considered in the next section, we use the expected value of the lower bound to the complete-data log-likelihood, given in \eqref{eq:mult_lb},
\begin{align}
	\label{eq:mult_M}
	\Exp_{q_m(\vtu)}\left[ \log Q_m(\vtu,\{\vc_i\}) \right] &= -\frac{1}{2} \Exp_{q_m(\vtu)}\left[ \vtu^T \left( \sum_{i=1}^P M_i \mtC_i \mA \mtC_i^T + \mI_{K(D-1)} \right) \vtu \right] \nn\\
	&\quad\quad + \Exp_{q_m(\vtu)}\left[ \vtu \right]^T \left( \sum_{i=1}^P M_i \mtC_i \mA \vth_i \right) + \text{Const.} \nn\\
	&= -\frac{1}{2} \Tr\left[ \left( \sum_{i=1}^P M_i \mtC_i \mA \mtC_i^T + \mI_{K(D-1)} \right) \left( \mPhi+\vphi\vphi^T \right) \right] \nn\\
	&\quad\quad + \vphi^T \sum_{i=1}^P M_i \mtC_i \mA \vth_i + \text{Const.},
\end{align}
where $\Tr(\cdot)$ is the trace of a matrix, and the expectation is taken with respect to $q_m(\vtu)$ (see \eqref{eq:mult_lb}).
To compute the expectation of the quadratic term we use the fact that $\Exp\left[ \vtu^T\mX\vtu \right] = \Exp\left[ \Tr(\vtu^T\mX\vtu) \right] = \Exp\left[ \Tr(\mX\vtu\vtu^T) \right] = \Tr\left( \mX \Exp[\vtu\vtu^T] \right)$. 

\subsection{Mixture Coefficients}
\label{sec:coef}

From \eqref{eq:vmf_M} and \eqref{eq:mult_M}, we estimate the mixture coefficients of the $i$-th hashtag as
\begin{align}
	\label{eq:coeff}
	\hat{\vc}_i &= \arg\max_{\vc_i} ~\Exp_{q_v(\vv_k)} \left[ \log Q_v(\{\vv_k\}, \{\vc_i\}) \right] + \Exp_{q_m(\vtu)}\left[ \log Q_m(\vtu,\{\vc_i\}) \right] \nn\\
	&= \left( \mB \frac{\vc_i}{\sum_{k=1}^K c_{ik}} \right)^T \kappa_i \sum_{n=1}^{N_i} \vw_{in} + \vphi^T M_i \mtC_i \mA \vth_i - \frac{1}{2} \Tr\left[ M_i \mtC_i \mA \mtC_i^T \left( \mPhi+\vphi\vphi^T \right) \right], 
\end{align}
where $\mB=[\vb_1 \ldots \vb_K]$ holds the posterior mean directions of the event geolocations (see \eqref{eq:vmf_E_mean}). 
From \eqref{eq:mult_mean},
\begin{align}
	\label{eq:phi_grad}
	\vphi^T M_i \mtC_i \mA \vth_i &= \vphi^T \left( \vz_i \otimes \vc_i \right) \nn\\
	&= \vc_i^T \mX \vz_i. 
\end{align}

Using the definitions of $\mA$ and $\mtC_i$, given in \eqref{eq:A_def} and \eqref{eq:C_def}, we write
\begin{align}
	M_i \mtC_i \mA \mtC_i^T &= \frac{M_i}{2} \mtC_i\mtC_i^T - \frac{M_i}{2D} (\mtC_i\vone_{D-1})(\mtC_i\vone_{D-1})^T \nn\\
	&= \mI_{D-1} \otimes \frac{M_i}{2} \vc_i \vc_i^T - \vone_{D-1}\vone_{D-1}^T \otimes \frac{M_i}{2D} \vc_i \vc_i^T. \nn
\end{align}
As a result, from \eqref{eq:Phi_fin},
\begin{align}
	\label{eq:Tr_Phi}
	\Tr\Big( M_i \mtC_i \mA \mtC_i^T \mPhi &\Big) = \Tr\left( \mI_{D-1} \otimes \frac{M_i}{2} \vc_i \vc_i^T\mF^{-1} - \vone_{D-1}\vone_{D-1}^T \otimes \frac{M_i}{2D} \vc_i \vc_i^T\mF^{-1} \right)  \nn\\
	& - \Tr\left( \vone_{D-1}\vone_{D-1}^T \otimes \frac{M_i}{2} \vc_i \vc_i^T\mDel - \vone_{D-1}\vone_{D-1}^T \otimes \frac{M_i(D-1)}{2D} \vc_i \vc_i^T\mDel \right) \nn\\
	&= \vc_i^T \left[ \frac{M_i(D-1)^2}{2D} \mF^{-1} - \frac{M_i(D-1)}{2D} \mDel \right] \vc_i.
\end{align}
Similarly, using \eqref{eq:X_def} we write
\begin{align}
	\label{eq:Tr_phi}
	\Tr\Big( M_i \mtC_i \mA \mtC_i^T \vphi\vphi^T &\Big) = \sum_{d=1}^{D-1} \Tr\left( \frac{M_i}{2} \vc_i \vc_i^T \vx_d\vx_d^T \right) - \sum_{d=1}^{D-1} \sum_{j=1}^{D-1} \Tr\left( \frac{M_i}{2D} \vc_i \vc_i^T \vx_d\vx_j^T \right)  \nn\\
	&= \vc_i^T \left[ \frac{M_i}{2} \sum_{d=1}^{D-1} \vx_d\vx_d^T - \frac{M_i}{2D} \sum_{d=1}^{D-1} \sum_{j=1}^{D-1} \vx_d\vx_j^T \right] \vc_i.
\end{align}
Substituting \eqref{eq:phi_grad}, \eqref{eq:Tr_Phi} and \eqref{eq:Tr_phi} in \eqref{eq:coeff} we have the following quadratic program
\begin{align}
	\label{eq:coeff_final}
	\hat{\vc}_i &= \arg\max_{\vc_i} ~ -\frac{1}{2} \vc_i^T \mGa_i \vc_i + \vc_i^T \vga_i \nn\\ 
	&\quad\quad \text{subject to}~~ c_{ik} \ge 0, ~k=1,\ldots,K, \\
	\mGa_i &= \frac{M_i(D-1)^2}{2D} \mF^{-1} - \frac{M_i(D-1)}{2D} \mDel + \frac{M_i}{2} \sum_{d=1}^{D-1} \vx_d\vx_d^T - \frac{M_i}{2D} \sum_{d=1}^{D-1} \sum_{j=1}^{D-1} \vx_d\vx_j^T \nn\\
	\vga_i &= \mB^T \kappa_i \sum_{n=1}^{N_i} \vw_{in} + \mX \vz_i, \nn
\end{align}
which can be efficiently solved using the interior point method. 

\begin{algorithm}[htb]
\caption{The proposed EM algorithm}
\label{alg:em}
\baselineskip=0.5cm
  \begin{algorithmic}[1]
    \STATE Input $\{ \bar{\vw}_i^{3 \times 1}=\sum_{n=1}^{N_i} \vw_{in}, \vh_i^{D \times 1} \}, ~i=1,\ldots,P$
    \STATE Initialize $\{\vc_i^{K \times 1},\kappa_i,s_k,\beta_k^{3 \times 1}\}, ~k=1,\ldots,K$    
    \WHILE{not converged}
    \STATE Compute posterior parameters $\{\vb_k, r_k, \tau_i\}$ for vMF as in \eqref{eq:vmf_E_mean}, \eqref{eq:vmf_E_conc}, \eqref{eq:tau}
    \STATE Update vMF parameters: \\
    $\kappa_i = \max\left\{ 0, \frac{3\tau_i-\tau_i^3}{1-\tau_i^2} \right\}, ~ s_k = \max\left\{ 0, \frac{3\vbe_k^T \vb_k-(\vbe_k^T \vb_k)^3}{1-(\vbe_k^T \vb_k)^2} \right\}, ~\vbe_k = \vb_k$
    \STATE Compute posterior parameters $\{\mPhi,\mX\}$ for multinomial as in \eqref{eq:Phi_fin}, \eqref{eq:phi_matrix}
    \STATE Update multinomial parameter $\vpsi_i = \mX^T \vc_i$
    \STATE Update mixture coefficients $\{\vc_i\}$ by solving \eqref{eq:coeff_final}
    \ENDWHILE
  \end{algorithmic}
\end{algorithm}

The resulting algorithm is summarized as Algorithm \ref{alg:em}.

{\subsection{Computational Complexity}
In the following theorem, we show that the computational complexity of Algorithm \ref{alg:em} scales linearly with each dimension of the problem. As a result, the proposed algorithm can be efficiently used for large datasets, as demonstrated in the following section.
\begin{theorem}
\label{thm:comp}
	At each iteration of the proposed EM algorithm, given by Algorithm \ref{alg:em}, the computational complexity linearly scales with the number of factors $K$, the number of hashtags 		$P$, and the number of words in the dictionary $D$, i.e., $O(KPD)$.	
\end{theorem}
\begin{proof}
First of all, note that typically $K \ll P \ll D$.
We start with the vMF E-step (line 4 in Algorithm \ref{alg:em}). The most expensive computation in the vMF E-step is the posterior mean direction $\vb_k$, given by \eqref{eq:vmf_E_mean}. Note that the sum of geolocation vectors $\sum_{n=1}^{N_i} \vw_{in}$ is computed offline once for each hashtag $i$; hence the number of geotagged tweets $N_i$ does not contribute to the computational complexity. Each $\vb_k$ has a computational complexity of $O(P)$. As a result, the computational complexity for the vMF E-Step is $O(KP)$. 
There is no expensive computation in the vMF M-step (line 5 in Algorithm \ref{alg:em}). \newline \indent
In the multinomial E-step (line 6 in Algorithm \ref{alg:em}), the computational complexity of the posterior covariance $\mPhi$ is $O(K^2P)$ due to the computation of $\mF$ and $\mDel$ (see \eqref{eq:Phi_fin}). The computation of the posterior mean $\mX$ in \eqref{eq:phi_matrix} has a complexity of $O(KPD)$ due to the multiplication of the matrices $\mC^{K\times P}$ and $\mZ^{P\times D}$. The computational complexity of the multinomial M-step (line 7 in Algorithm \ref{alg:em}) is $O(KD)$.\newline \indent
Finally, the complexity for updating the coefficients $\{\vc_i\}$ (line 8 in Algorithm \ref{alg:em}) is $O(KPD)$ since for each $\vga_i$, given by \eqref{eq:coeff_final}, the complexity is $O(KD)$ due to the multiplication $\mX\vz_i$. In \eqref{eq:coeff_final}, the matrices $\{\mGa_i\}$ entail the complexity of $O(K^2D)$ due to the computation of $\sum_{d=1}^{D-1} \vx_d\vx_d^T = \mX\mX^T$, which is computed once and used for each $\mGa_i$. Note that $\sum_{d=1}^{D-1} \sum_{j=1}^{D-1} \vx_d\vx_j^T = \left(\sum_{d=1}^{D-1} \vx_d \right) \left(\sum_{d=1}^{D-1} \vx_d \right)^T$ is nothing but the outer product of the sum vector $\sum_{d=1}^{D-1} \vx_d$, requiring $O(KD)$ computations for the sum and $O(K^2)$ computations for the outer product. In solving the constrained quadratic program given in \eqref{eq:coeff_final} for each coefficient vector $\vc_i$, the number of iterations, in practice, is bounded by a constant; and in each iteration linear algebra operations in the $K$-dimensional space are performed. Hence, the overall complexity does not exceed $O(KPD)$. Note also that each $\vc_i$ can be updated in parallel. 
\end{proof}
}

\section{Experiments}
\label{sec:exp}

We have tested the proposed algorithm on a Twitter dataset from August 2014 obtained from the Twitter stream API at gardenhose level access. It spans the whole month, and includes a random sample of 10 \% of all tweets from all over the world. We consider about 30 million geotagged tweets, among which around 3 million use approximately 1 million unique hashtags. We have organized the data in terms of hashtags. That is, each unique hashtag is an instance with bag-of-words and geolocation features. The rarely used hashtags and the hashtags with small geographical distribution are filtered out, leaving us with 13000 hashtags ($P=13000$), and a dictionary of $67000$ significant words ($D=67000$). The number of geotags, $N_i$, for hashtags varies from $2$ to $71658$; and the number of words, $M_i$, varies from $10$ to $426892$.
%The number of events, $K$, is selected using a recursive adaptive procedure. We start with a large number (e.g., $K=40$), and after a certain number of iterations, automatically remove the uninformative events that are not associated with any hashtag, i.e., $\{c_{ik}\}_i$ are all small for a given $k$. 

We run the algorithm in a hierarchical manner with a small number of events in each round (e.g., $K=10$). In each round, the hashtags that localize well in an event with a dominant mixture coefficient are pruned, and the remaining hashtags are further processed in the next round. In other words, in the next round, we zoom into the previously unexplored sections of the data to discover new events. We also zoom into the broadly discovered events to find specific events. For example, in the first round, we have discovered popular events such as the Ice Bucket Challenge and the Robin Williams' death, and also generic events for the British and Asian hashtags (top figure in Fig. \ref{fig:factors}). In the following round, separately processing the generic events and the non-localized data we have identified further specific events such as the Ferguson unrest and the USA national basketball team for the FIBA world championship (bottom figure in Fig. \ref{fig:factors}). Specifically, we have identified an Indian event about a Hindu religious leader in jail, and a British event about the Commonwealth Games in the generic Asian and British events, respectively. In Fig. \ref{fig:factors}, it is seen also that the previously found Ice Bucket Challenge event has decomposed into a local event and a global event in the second round. It is seen in Fig. \ref{fig:factors} that the proposed algorithm successfully finds interesting events in terms of both topic and geolocation. The geographical distribution of the tweets that use hashtags associated with the events about the Commonwealth Games and the Hindu religious leader are depicted in Fig. \ref{fig:glasgow}. Similarly, Fig. \ref{fig:robin} illustrates the geographical distributions of the tweets that use hashtags about the death of Robin Williams. The geolocations of the tweets that are shown in Fig. \ref{fig:glasgow} and Fig. \ref{fig:robin} are consistent with the corresponding events. As expected, the tweets that mention Robin Williams are well distributed around the world with a center in the USA, whereas the tweets about the Commonwealth Games are sent only from the Commonwealth countries, and the tweets about the Hindu leader are only from India.

\begin{figure}[thb]
\includegraphics[width=\textwidth]{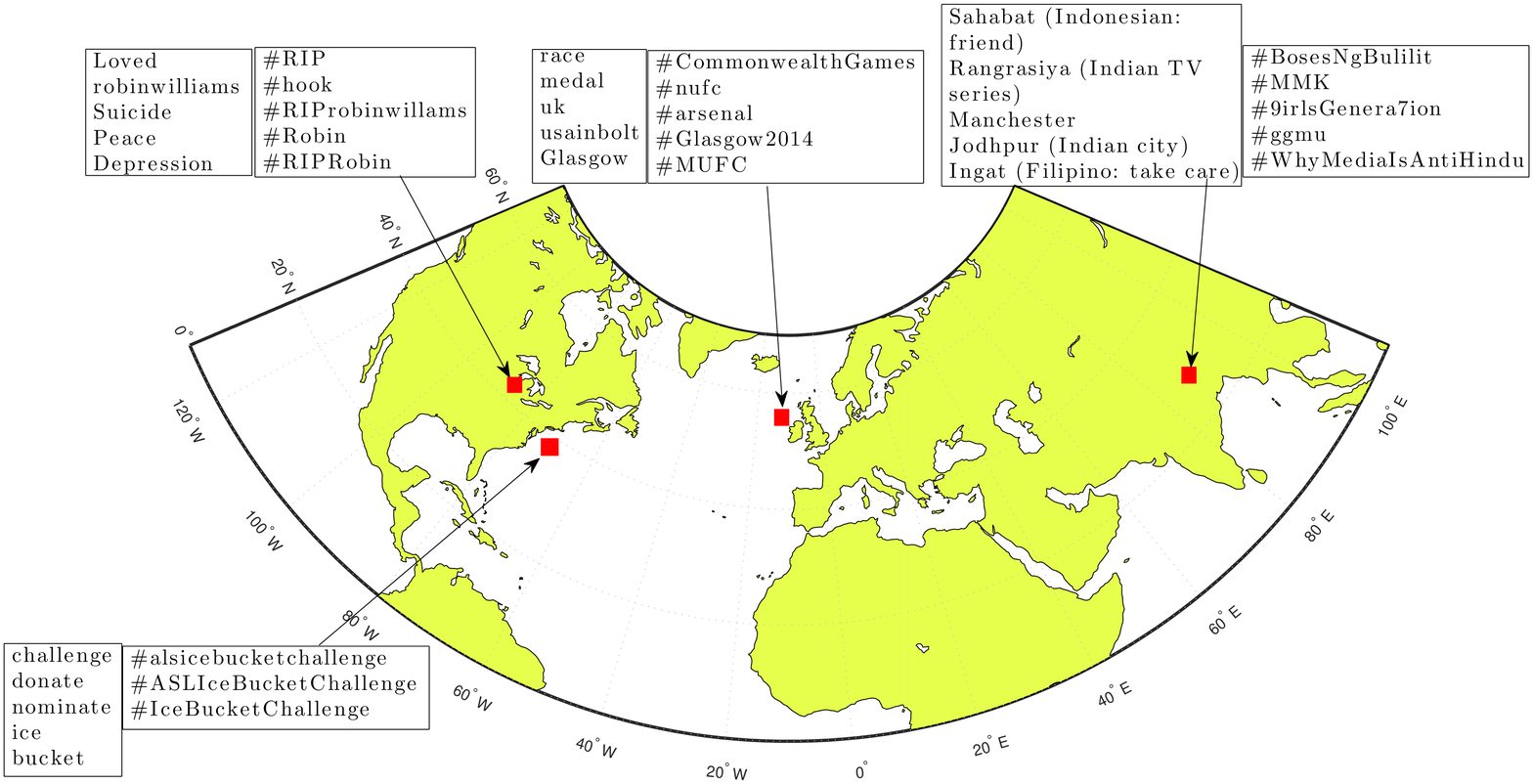}
\includegraphics[width=\textwidth]{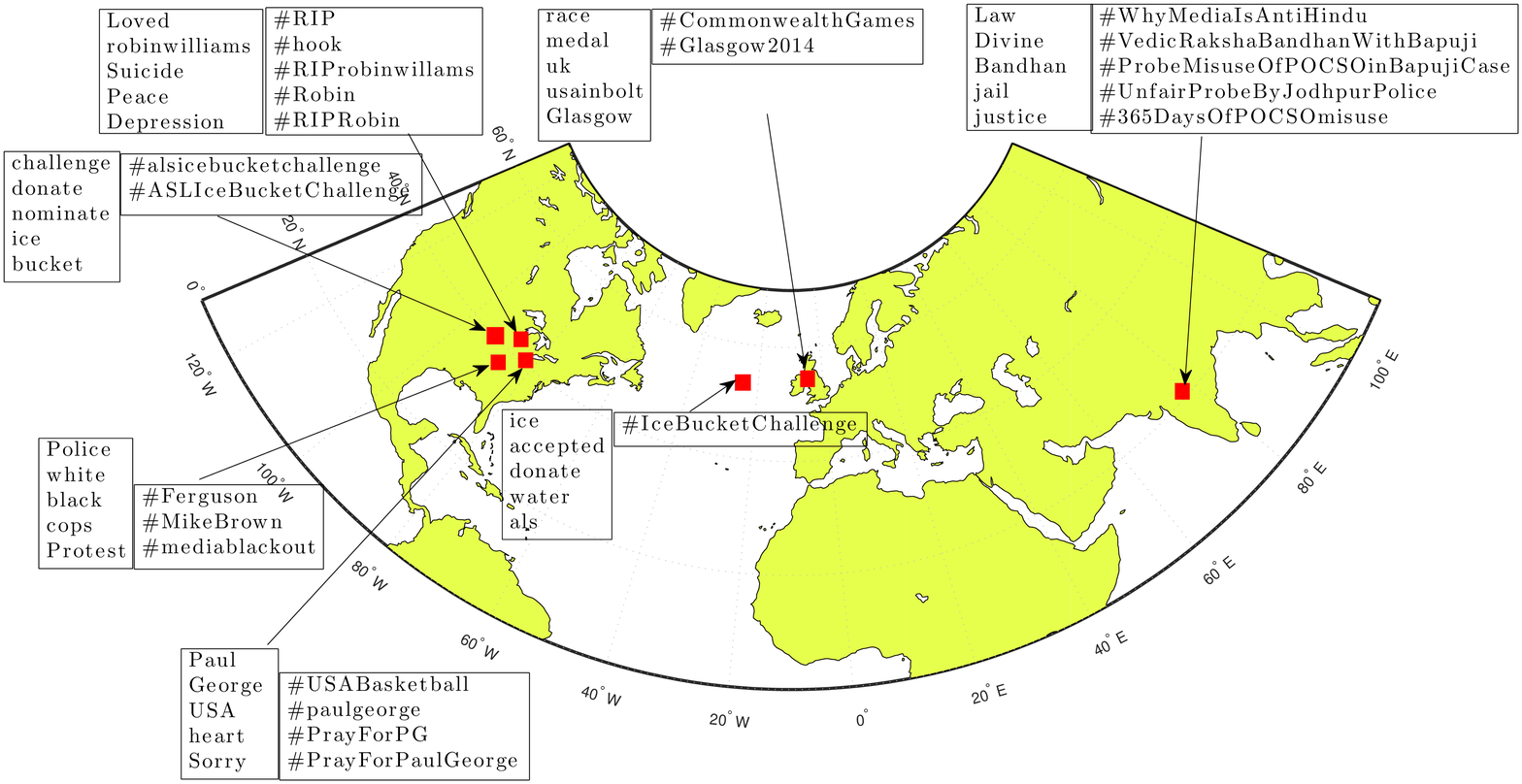}
\caption{Some events discovered in the first round of the algorithm (top). Some specific events discovered after two rounds (bottom). Dominant hashtags and words used for the events, as well as their mean geolocations are displayed.}
\label{fig:factors}       % Give a unique label
\end{figure}

\begin{figure}[thb]
\includegraphics[width=\textwidth]{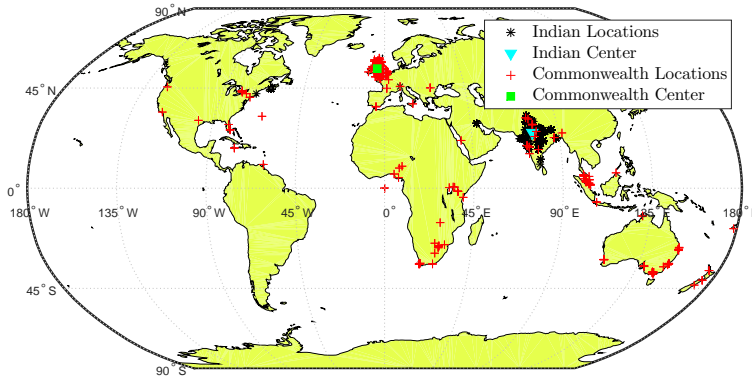}
\caption{Geographical distribution of the tweets the algorithm associates with the Commonwealth Games and the Hindu religious leader in jail.}
\label{fig:glasgow}       % Give a unique label
\end{figure}

\begin{figure}[thb]
\includegraphics[width=\textwidth]{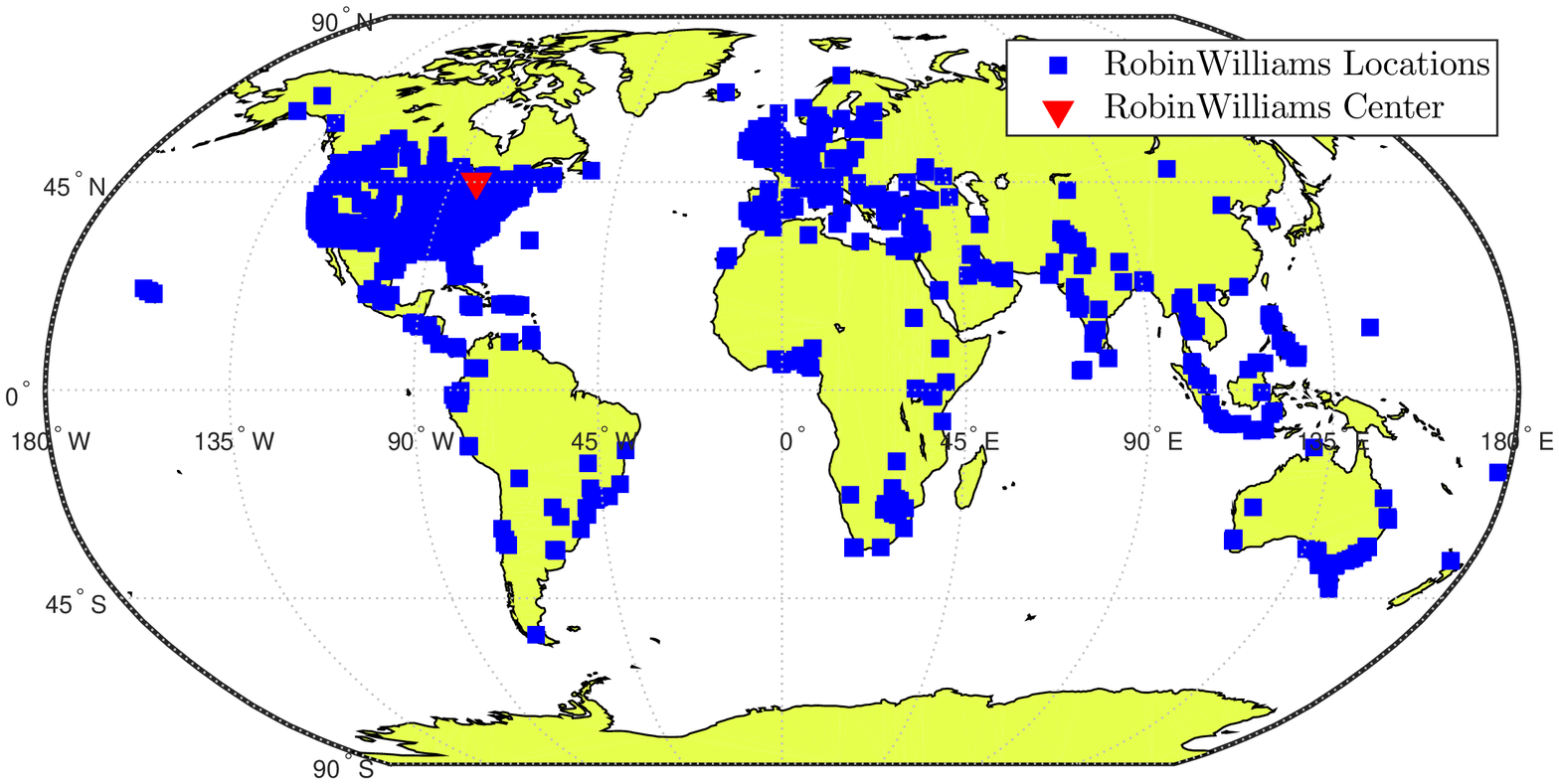}
\caption{Geographical distribution of the tweets the algorithm associates with the death of Robin Williams.}
\label{fig:robin}       % Give a unique label
\end{figure}

Next, as an application, we cluster the hashtags based on the mixture coefficients $\vc_i$. 
A sample result using k-means and multidimensional scaling (MDS) is shown in Fig. \ref{fig:cluster}. 
{For this example, we have chosen a small subset of the dataset with $314$ hashtags and annotated each hashtag with an event name, which resulted in $7$ events as shown in Fig. \ref{fig:cluster}. This hand annotation provides us a ground truth (color and shape coded by markers), which we compare with the clustering result from the proposed algorithm (shown by ellipses). Running the algorithm with $K=7$ we get a $7$-dimensional coefficient vector $\vc_i$ for each hashtag $i$ \footnote{The issue of unknown number of events can be handled using the silhouette values.}. We then cluster the vectors $\{\vc_i\}$ into $7$ groups using the k-means technique. The rand index and the adjusted rand index between the clustering result and the ground truth are $0.9807$ and $0.9539$, respectively. The rand index is a similarity measure between two clusterings. It takes a value between $0$ and $1$, where $0$ indicates no agreement, and $1$ indicates perfect agreement. Hence, from the rand index result and also Fig. \ref{fig:cluster}, it is seen that the proposed algorithm can be used to effectively cluster multimodal big datasets.}

\begin{figure}[thb]
\includegraphics[width=\textwidth]{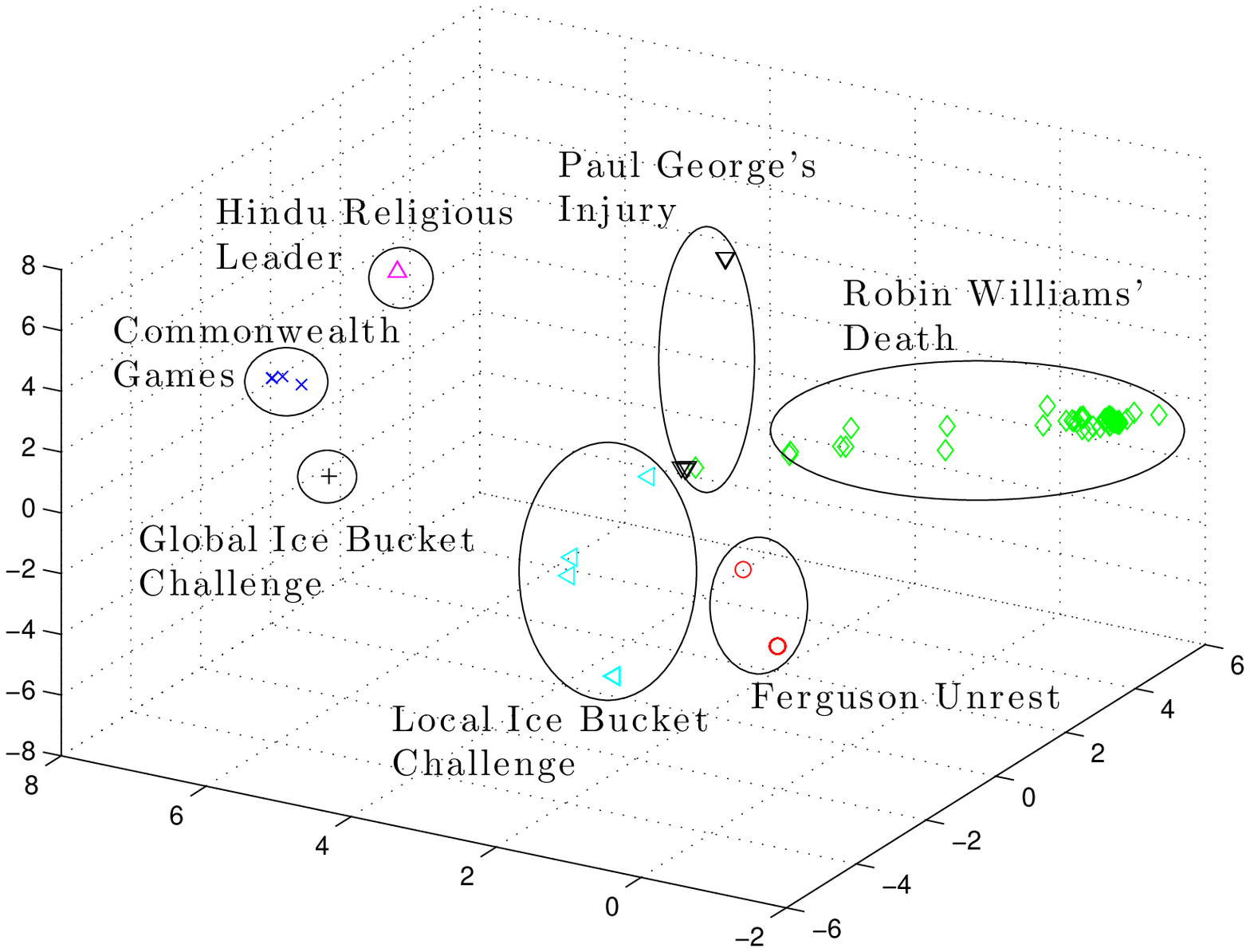}
\caption{Clustering via k-means and MDS based on the mixture coefficients. The ground truth is shown with the color and shape of the markers; and the clustering result is shown with the ellipses.}
\label{fig:cluster}       % Give a unique label
\end{figure}

{Finally, for the same subset used in Fig. \ref{fig:cluster}, we measure, in Fig. \ref{fig:good}, how well our model fits the data by comparing the likelihood values of the geolocation data under the vMF model given by the proposed algorithm with the likelihood values under the individual vMF models separately fitted for each hashtag. The individual vMF models for each hashtag provide us a baseline for comparison since this is what one would typically do to fit the geolocation data for each hashtag without the goal of event detection. Fig. \ref{fig:good} shows that the proposed event detection algorithm fits the geolocation data as well as the baseline data fitting technique.}

\begin{figure}[thb]
\includegraphics[width=\textwidth]{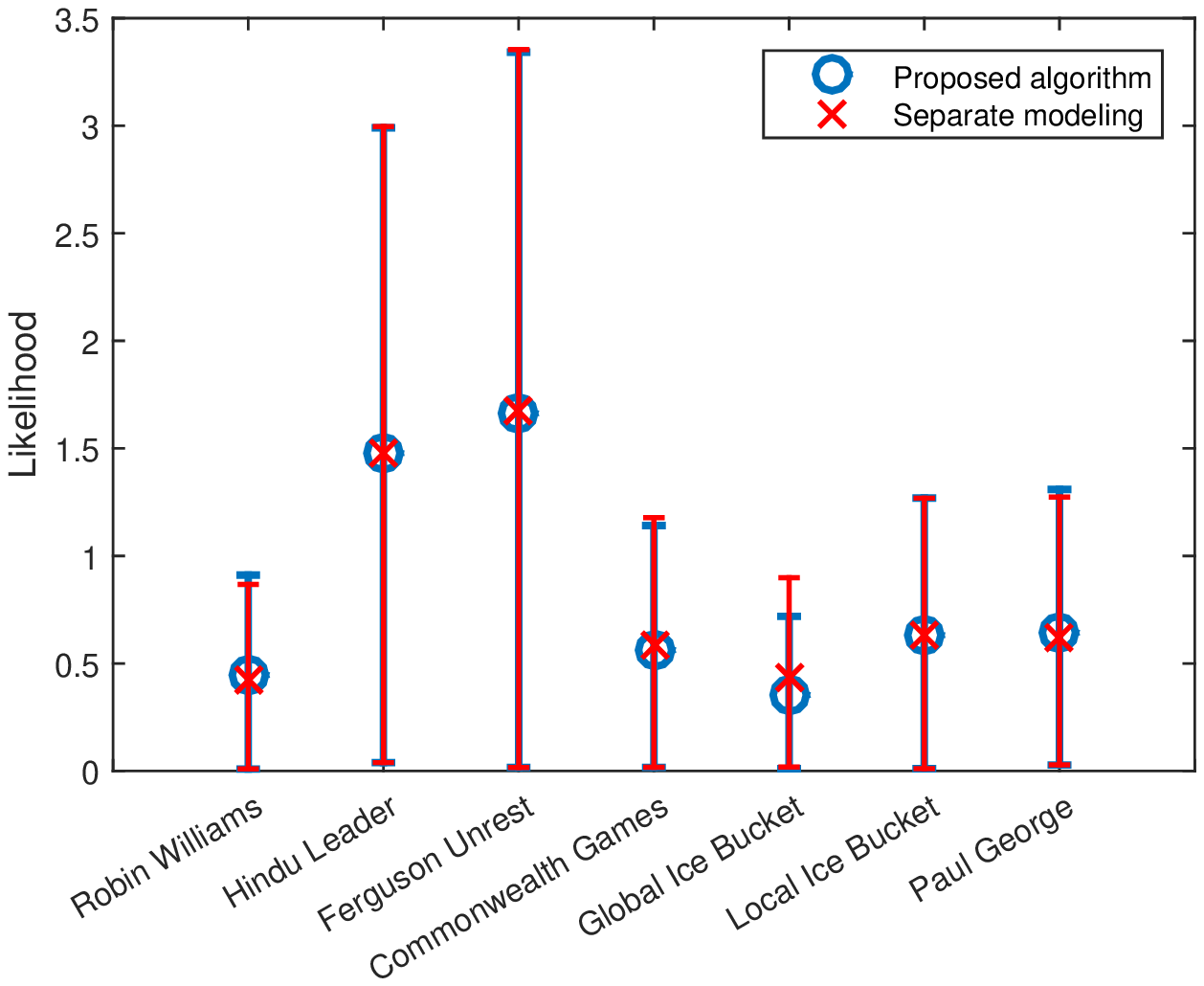}
\caption{Goodness of fit plot for the proposed event detection algorithm and the individual vMF models separately fitted to the geolocation data for each hashtag. Note that the separate modeling does not perform event detection; it is the baseline technique for data fitting. It is seen that the proposed algorithm not only detects events (see the previous figures), but also fits the geolocation data as well as the baseline technique. For each event, the mean likelihood values are shown with markers (circle and cross), and the $95\%$ confidence intervals are shown with bars.}
\label{fig:good}       % Give a unique label
\end{figure}

\section{Conclusion}
\label{sec:conc}

We have treated the event detection problem in a multimodal Twitter hashtag network. Utilizing the bag-of-words and the geotags from related tweets as the features for hashtags we have developed a variational EM algorithm to detect events according to a generative model. The computational complexity of the proposed algorithm has been simplified such that it is viable for big datasets. A hierarchical version of the proposed algorithm has been tested on a Twitter dataset with 13000 hashtags from August 2014. By pruning data in each round multi-resolution events (higher in each round) have been learned. Significant events, such as Robin Williams' death, and the Ice Bucket Challenge, as well as some generic events, such as the British and the Asian hashtags, have been learned in the first round. Later in the second round, new specific events have been discovered within the generic events. We have also successfully clustered a set of hashtags using the detected events. {In addition to event detection, we have shown that the proposed algorithm fits the geolocation data as well as the baseline data fitting technique which separately models each hashtag. The proposed algorithm is justified by the remarkable clustering and goodness of fit results, and the low computational complexity that linearly scales with the number of events, number of hashtags, number of tweets, and number of words in the dictionary.}
%The number of events has been automatically set by removing the uninformative events that are not associated with any hashtag after a certain number of iterations. 

%\begin{acknowledgements}
%\end{acknowledgements}



\end{document}